\documentclass[UKenglish]{lmcs}
\pdfoutput=1

\usepackage{lastpage}
\renewcommand{\dOi}{14(3:14)2018}
\lmcsheading{}{1--\pageref{LastPage}}{}{}%
{Oct.~12,~2017}{Sep.~03,~2018}{}

\usepackage{microtype}
\usepackage{amsmath,amssymb,amsfonts,mathrsfs}
\usepackage[all,2cell]{xy}
\usepackage{graphicx}
\usepackage{hyperref}
\usepackage{anyfontsize}

\newcommand{\tr}[1]{\xrightarrow{#1}}
\newcommand{\tl}[1]{\xleftarrow{#1}}
\newcommand{\maps}{{\colon}}

\def \catC {\mathcal{C}}
\def \subc {\mathcal{A}}
\newcommand{\op}[1]{{#1}^{\scriptscriptstyle op}}
\newcommand{\Span}[1]{\mathsf{Span}(#1)}
\newcommand{\Spanc}[1]{\mathsf{Span}(#1)}
\newcommand{\Cospan}[1]{\mathsf{Cospan}(#1)}
\newcommand{\Rel}[1]{\mathsf{Rel}(#1)}
\newcommand{\Corel}[1]{\mathsf{Corel}(#1)}
\def \PROP {\mathbf{Prop}} % category of PROPs
\def \CAT {\mathbf{Cat}} % category of small categories
\def \SET {\mathsf{Set}} % category of sets
\def \poi {\,\ensuremath{;}\,} % sequential composition
\def \tns {\ensuremath{\oplus}} % monoidal product
\def \id {\mathit{id}} %identity map

\newcommand\pair[1]{\langle #1 \rangle} %pairing of a cartesian product
\newcommand\copair[1]{[#1]} %copairing of a coproduct
\newcommand\ord[1]{\overline{#1}}

\def \CospanToCorel {\Gamma}
\def \SpanToCorel {\Pi}
\newcommand \+[1] {+_{\scriptscriptstyle \lvert #1 \rvert}}% coproduct taking only one set of objects

\renewcommand{\fact}[2]{({#1} , {#2})}
\newcommand{\Cepi}{\mathcal{E}} % epi subcategory
\newcommand{\Cmono}{\mathcal{M}} % mono subcategory
\newcommand{\Ciso}{\mathcal{I}} % iso subcategory

\newcommand{\Triv}{\mathbf{1}} % trivial PROP
\newcommand{\F}{\mathsf{F}} % PROP of Functions
\newcommand{\Inj}{\mathsf{In}}  % the prop of injections
\newcommand{\Surj}{\mathsf{Su}} % the prop of surjections
\newcommand{\ER}{\mathsf{ER}}  % prop of equivalence relations
\newcommand{\PER}{\mathsf{PER}}  % prop of partial equiv relations
\def \PF {\mathsf{PF}} % partial functions
\newcommand \Vect[1] {\mathsf{Vect}_{\scriptscriptstyle #1}}
\newcommand \vect[1] {\mathsf{VECT}_{\scriptscriptstyle #1}}
\newcommand \SV[1] {\mathsf{SV}_{\scriptscriptstyle #1}}
\def \PID {\mathsf{R}}
\def \field {\mathsf{k}} %field
\def \frPID {\mathsf{k}} %field of fractions
\def \Z {\mathsf{Z}}

\newcommand{\fmod}[1]{\mathsf{FMod}_{\scriptscriptstyle #1}}
\newcommand{\mfmod}[1]{\mathsf{MFMod}_{\scriptscriptstyle #1}} % split monos module homomorphsims

\newcommand{\pullbacktop}[4]{%
{#1} \ar@/_/[ddr]_{#4} \ar@/^/[drr]^{#2}%
\ar@{.>}[dr]|-{#3} \\}

\newcommand{\pullbackcorner}[1][dr]{\save*!/#1+1.5pc/#1:(1,-1)@^{|-}\restore}

\newcommand{\cgr}[2][scale=0.45]{\raisebox{0.1em}{\begingroup
\setbox0=\hbox{\includegraphics[#1]{graffles/#2}}%
\parbox{\wd0}{\box0}\endgroup}}

\newcommand\Bmult{\lower5pt\hbox{$\includegraphics[width=20pt]{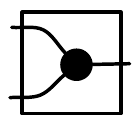}$}}
\newcommand\Bcomult{\lower5pt\hbox{$\includegraphics[width=20pt]{graffles/Bcomult.pdf}$}}
\newcommand\Bunit{\cgr[height=10pt]{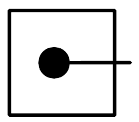}}
\newcommand\Bcounit{\cgr[height=10pt]{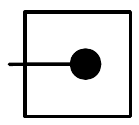}}
\newcommand\EmptyDiag{\cgr[height=10pt]{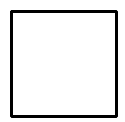}}
\newcommand\IdDiag{\cgr[height=12pt]{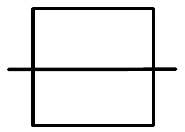}}
\newcommand\scalar{\cgr[height=12pt]{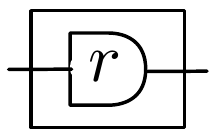}}
\newcommand\scalarop{\cgr[height=12pt]{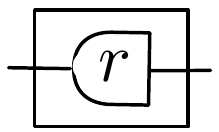}}

\theoremstyle{definition}
\newtheorem{ass}[]{Assumption}
\newtheorem{defn}[]{Definition}
\newtheorem{example}[]{Example}[section]
\newtheorem{rmk}[]{Remark}[section]

\begin{document}

% Author macros::begin %%%%%%%%%%%%%%%%%%%%%%%%%%%%%%%%%%%%%%%%%%%%%%%%
\title[Universal Constructions for (Co)Relations]{Universal Constructions for
(Co)Relations: \\ categories, monoidal categories, and props}
\titlecomment{{\lsuper*}Extended version of \textsl{A Universal Construction for
(Co)Relations}, CALCO 2017.}

\author[B.~Fong]{Brendan Fong}
\address{Massachusetts Institute of Technology, United States of America}
\email{bfo@mit.edu}
\thanks{We thank David Spivak for useful conversations. BF acknowledges support
from the Basic Research Office of the USA ASDR\&E through ONR N00014-16-1-2010,
AFOSR FA9550-14-1-0031, and AFOSR FA9550-17-1-0058. FZ acknowledges support from EPSRC grant n. EP/R020604/1.
}

  \author[F.~Zanasi]{Fabio Zanasi}
\address{University College London, United Kingdom}
\email{F.Zanasi@cs.ucl.ac.uk}


% Author macros::end %%%%%%%%%%%%%%%%%%%%%%%%%%%%%%%%%%%%%%%%%%%%%%%%%

\begin{abstract}
Calculi of string diagrams are increasingly used to present the syntax and
algebraic structure of various families of circuits, including signal flow
graphs, electrical circuits and quantum processes. In many such approaches, the
semantic interpretation for diagrams is given in terms of relations or
corelations (generalised equivalence relations) of some kind. In this paper we
show how semantic categories of both relations and corelations can be
characterised as colimits of simpler categories. This modular perspective is
important as it simplifies the task of giving a complete axiomatisation for
semantic equivalence of string diagrams. Moreover, our general result unifies
various theorems that are independently found in literature and are relevant for program semantics, quantum computation and control theory.
\end{abstract}

\maketitle

\section{Introduction}

Network-style diagrammatic languages appear in diverse fields as a tool to
reason about computational models of various kinds, including signal processing
circuits, quantum processes, Bayesian networks and Petri nets, amongst many
others. In the last few years, there have been more and more contributions
towards a uniform, formal theory of these languages which borrows from the
well-established methods of programming language semantics. A significant
insight stemming from many such approaches is that a \emph{compositional}
analysis of network diagrams, enabling their reduction to elementary components,
is more effective when system behaviour is thought of as a \emph{relation} instead of a function. 

A paradigmatic case is the one of signal flow graphs, a foundational structure
in control theory: a series of recent
works~\cite{Bonchi2014b,BaezErbele-CategoriesInControl,Bonchi2015,BonchiSZ17,Fong2015}
gives this graphical language a syntax and a semantics where each signal flow
diagram is interpreted as a subspace (a.k.a. linear relation) over streams. The
highlight of this approach is a sound and complete axiomatisation for semantic
equivalence: what is of interest for us is how this result is achieved
in~\cite{Bonchi2014b}, namely through a \emph{modular} account of the domain of
subspaces. The construction can be studied for any field $\field$: one
considers the prop\footnote{A prop is a symmetric monoidal category with
objects the natural numbers \cite{MacLane1965}. It is the typical setting for
studying both the syntax and the semantics of network diagrams.} $\SV{\field}$
whose arrows $n \to m$ are subspaces of $\field^n \times \field^m$, composed as
relations. As shown in in~\cite{interactinghopf,ZanasiThesis}, $\SV{\field}$
enjoys a universal characterisation: it is the pushout (in the category of
props) of props of \emph{spans} and of \emph{cospans} over $\Vect{\field}$, the prop
with arrows $n \to m$ the linear maps $\field^n \to \field^m$:
\begin{equation}\label{intro:linearcube}
\begin{aligned}
    \xymatrix@R=18pt{
    {\Vect{\field} + \op{\Vect{\field}}} \ar[r]
    \ar[d]& {\Span{\Vect{\field}}}
    \ar[d] \\
    {\Cospan{\Vect{\field}}} \ar[r] & {\SV{\field}}. \pullbackcorner
    }
  \end{aligned}
\end{equation}
In linear algebraic terms, the two factorisation properties expressed by
\eqref{intro:linearcube} correspond to the representation of a subspace in terms
of a basis (span) and the solution set of a system of linear equations (cospan).
Most importantly, this picture provides a roadmap towards a complete
axiomatisation for $\SV{\field}$: one starts from the domain $\Vect{\field}$ of
linear maps, which is axiomatised by the equations of Hopf algebras, then combines
it with its opposite $\op{\Vect{\field}}$ via two distributive laws of
props~\cite{Lack2004a},  one yielding an axiomatisation for
$\Span{\Vect{\field}}$ and the other one for $\Cospan{\Vect{\field}}$. Finally,
merging these two axiomatisations yields a complete axiomatisation for $\SV{\field}$,
called the theory of interacting Hopf algebras~\cite{interactinghopf,ZanasiThesis}.

\medskip

It was soon realised that this modular construction was of independent interest, and perhaps evidence of a more general phenomenon. In~\cite{Zanasi16} it is shown that a similar construction could be used to characterise the prop $\ER$ of equivalence relations, using as ingredients $\Inj$, the prop of injections, and $\F$, the prop of total functions. The same result is possible by replacing equivalence relations with partial equivalence relations and functions with partial functions, forming a prop $\PF$. In both cases, the universal construction yields a privileged route to a complete axiomatisation, of $\ER$ and of $\PER$ respectively~\cite{Zanasi16}. 
\begin{equation}\label{intro:cartesiancubes}
\begin{aligned}
    \xymatrix@R=18pt{
    {\Inj + \op{\Inj}} \ar[r]
    \ar[d]& {\Span{\Inj}}
    \ar[d] \\
    {\Cospan{\F}} \ar[r] & {\ER} \pullbackcorner
    }
    \qquad \qquad
    \xymatrix@R=18pt{
    {\Inj + \op{\Inj}} \ar[r]
    \ar[d]& {\Span{\Inj}}
    \ar[d] \\
    {\Cospan{\PF}} \ar[r] & {\PER} \pullbackcorner.
    }
  \end{aligned}
\end{equation}
Even though a pattern emerges, it is certainly non-trivial: for instance, if one naively mimics the linear case~\eqref{intro:linearcube} in the attempt of characterising the prop of relations, the construction collapses to the terminal prop $\Triv$.
\begin{equation}\label{intro:cartesiancubecollapses}
\begin{aligned}
    \xymatrix@R=18pt{
    {\F + \op{\F}} \ar[r]
    \ar[d] & {\Span{\F}} \ar[d]^{} \\
    {\Cospan{\F}} \ar[r]_-{} & {\Triv}\pullbackcorner
    }
  \end{aligned}
\end{equation}

More or less at the same time, diagrammatic languages for various families of
circuits, including linear time-invariant dynamical systems \cite{Fong2015},
were analysed using so-called \emph{corelations}, which are generalised
equivalence relations~\cite{Fon16,CF,Fon17,Fong2013}. Even though they were not
originally thought of as arising from a universal construction like the examples above, corelations still follow a modular recipe, as they are expressible as a quotient of $\Cospan{\catC}$, for some prop $\catC$. Thus by analogy we can think of them as yielding one half of the diagram
\begin{equation}\label{intro:corelations}
\begin{aligned}
    \xymatrix@R=18pt{
    {\catC + \op{\catC}}
    \ar[d] &  \\
    {\Cospan{\catC}} \ar[r]_-{} & {\Corel{\catC}}.
    }
  \end{aligned}
\end{equation}
In this paper we clarify the situation by giving a unifying perspective for all these constructions. We prove a general result, which
\begin{itemize}
\item implies \eqref{intro:linearcube} and \eqref{intro:cartesiancubes} as special cases;
\item explains the failure of \eqref{intro:cartesiancubecollapses};
\item extends \eqref{intro:corelations} to a pushout recipe for corelations.
\end{itemize}
More precisely, our theorem individuates sufficient conditions for characterising the category $\Rel{\catC}$ of $\catC$-relations as a pushout. A dual construction yields the category $\Corel{\catC}$ of $\catC$-corelations as a pushout. For the case of interest when $\catC$ is a prop, the two constructions look as follows.
\[
\begin{aligned}
    \xymatrix@R=18pt{
      {\subc + \op{\subc}} \ar[r] \ar[d] & {\Spanc{\catC}}
      \ar[d] \\
    {\Cospan{\subc}} \ar[r] & {\Rel{\catC}} \pullbackcorner.
    }
 \qquad \qquad 
    \xymatrix@R=18pt{
      {\subc + \op{\subc}} \ar[r] \ar[d] & {\Spanc{\subc}}
      \ar[d] \\
    {\Cospan{\catC}} \ar[r] & {\Corel{\catC}} \pullbackcorner.
    }
  \end{aligned}
\]
The variant ingredient $\subc$ is a subcategory of $\catC$. In order to make the constructions possible, $\subc$ has to satisfy certain requirements in relation with the factorisation system $(\Cepi, \Cmono)$ on $\catC$ which defines $\catC$-relations (as jointly-in-$\Cmono$ spans) and $\catC$-corelations (as jointly-in-$\Cepi$ cospans). For instance, taking $\subc$ to be $\catC$ itself succeeds in \eqref{intro:linearcube} (and in fact, for any abelian $\catC$), but fails in \eqref{intro:cartesiancubecollapses}. 

Besides explaining existing constructions, our result opens the lead for new applications. In particular, we observe that under mild conditions the construction of relations lifts to the category $\catC^T$ of $T$-algebras for a monad $T \colon \catC \to \catC$, and dually for corelations and comonads. We leave the exploration of this and other ramifications for future work.

\subsection*{Synopsis} 
Section~\ref{sec:corelations} introduces factorisation systems and
(co)relations, and shows the subtleties of mapping spans into corelations in a
functorial way.  Section~\ref{sec:theorem} states our main result and some of
its consequences.  We first formulate the construction for categories
(Theorem~\ref{thm:corelations}), and then for props
(Theorem~\ref{thm:corelationsPROPs}), which are our prime object of interest in
applications. Section~\ref{sec:examples} is devoted to show various instances of
our construction. We illustrate the case of equivalence relations, of partial
equivalence relations, of subspaces, of linear corelations, and finally of
(co)relations of (co)algebras over a (co)monad. Section~\ref{sec:conclusion}
summarises our contribution and looks forward to further work. Appendix
\ref{sec:proof} unfolds the full proof of the main result,
Theorem~\ref{thm:corelations}, and Appendix \ref{app:spansPROP} does the same
for the extension to monoidal categories. 

This work is based on the conference paper \cite{FZ-calco17}, which has been
expanded to include all the omitted proofs, the monoidal case of the main
construction (Section \ref{sec:monoidal}), and new explanations and examples,
including two additional case studies: the non-example of plain relations
(Section \ref{ex:plainrelations}) and the example of corelations for coalgebras
(Section \ref{ssec.comonadcoalg}). Also, the generic construction of relations
for algebras (Section~\ref{prop.algebras}) has been amended and extended.

\subsection*{Conventions} 
%\medskip \textbf{Conventions} 
We write $f \poi g$ for composition of $f \colon X \to Y$ and $g \colon Y \to Z$ in a category $\catC$. It will be sometimes convenient to indicate an arrow $f \colon X \to Y$ of $\catC$ as $X\tr{f \in \catC}Y$ or also $\tr{\in \catC}$, if names are immaterial. Also, we write $X \tl{f \in \catC} Y$ for an arrow $X \tr{f \in \op{\catC}} Y$. We use $\tns$ for the monoidal product in a monoidal category, with unit object $I$. Monoidal categories and functors will be strict when not stated otherwise.

\section{(Co)relations} \label{sec:corelations}
In this section we review the categorical approach to relations, which
originates in the work of Hilton \cite{Hil66} and Klein \cite{Kle70}. This
categorical approach is based on the observation that, in $\SET$, relations are
the jointly mono spans---that is, spans $X \tl{f}N \tr{g} Y$ such that the
pairing $\langle f,g \rangle\colon N \to X \times Y$ is an injection. We
introduce in parallel the dual notion, called corelations \cite{Fon16}. These
are generalisations of jointly epi cospans---that is, cospans $X \tr{f} N \tl{g}
Y$ such that $[f,g]\colon X+Y \to N$ is surjective---and can be seen as an
abstraction of the concept of equivalence relation.

\begin{defn}
  A {\bf factorisation system} $(\Cepi,\Cmono)$ in a category
  $\mathcal C$ comprises subcategories $\Cepi$, $\Cmono$ of $\mathcal
  C$ such that
  \begin{enumerate}
    \item $\Cepi$ and $\Cmono$ contain all isomorphisms of $\mathcal
      C$.
    \item  every morphism $f \in \mathcal C$ admits a factorisation $f=e;m$, $e \in \Cepi$, $m \in \Cmono$.
\item given $f,f'$, with factorisations $f = e;m$, $f' = e';m'$ of the above
  sort, for every $u$, $v$ such that $f;v = u;f'$
  there exists a unique $s$ making the following diagram commute.
    \[
    \xymatrixcolsep{2pc}
    \xymatrixrowsep{2pc}
    \xymatrix{
      \ar[r]^e \ar[d]_u & \ar[r]^m \ar@{-->}[d]^{\exists! s} &  \ar[d]^v \\
       \ar[r]_{e'}& \ar[r]_{m'} &
    }
  \]
  \end{enumerate}
\end{defn}

\begin{defn}
  Given a category $\catC$, we say that a subcategory $\subc$ is {\bf stable
  under pushout} if for every pushout square 
  \[
    \xymatrix{
      \ar[r]^a \ar[d] & \ar[d]\\
      \ar[r]^f & {}%\pullbackcorner
    }
  \]
  such that $a \in \subc$, we also have that $f \in \subc$. Similarly, we say
  that $\subc$ is {\bf stable under pullback} if for every pullback square
  labelled as above $f \in \subc$ implies $a \in \subc$.

  A factorisation system $(\Cepi,\Cmono)$ is {\bf stable} if $\Cepi$ is stable
  under pullback, {\bf costable} if $\Cmono$ is stable under pushout, and {\bf
  bistable} if it is both stable and costable.
\end{defn}

Examples of bistable factorisation systems include the trivial factorisation
systems $(\mathcal I_{\catC},\catC)$ and $(\catC,\mathcal I_{\catC})$ in any category
$\catC$, where $\mathcal I_{\catC}$ is the subcategory containing exactly the
isomorphisms in $\catC$, the epi-mono factorisation system in any topos, or the
epi-mono factorisation system in any abelian category. Stable factorisation
systems include the (regular epi, mono) factorisation system in any regular
category, such as any category monadic over $\SET$. Dually, costable
factorisation systems include the (epi, regular mono) factorisation system in any
coregular category, such as the category of topological spaces and continuous
maps.

\begin{defn}
Given a category $\catC$ with pushouts, the category $\Cospan{\catC}$ has the same objects as $\catC$ and arrows $X \to Y$ isomorphism classes of cospans $X \tr{f}  \tl{g} Y$ in $\catC$. The composite of $X \tr{f}  \tl{g} Y$ and $Y \tr{h}  \tl{i} Z$ is obtained by taking the pushout of $\tl{g} \tr{h}$.

Given a category $\catC$ with pullbacks, the category $\Span{\catC}$ has the same objects as $\catC$ and arrows $X \to Y$ isomorphism classes of spans $X \tl{f}  \tr{g} Y$ in $\catC$. The composite of $X \tl{f}  \tr{g} Y$ and $Y \tl{h}  \tr{i} Z$ is obtained by taking the pullback of $\tr{g}\tl{h}$.
\end{defn}

When $\catC$ also has a (co)stable factorisation system, we may define a category
of (co)relations with respect to this system.

\begin{defn}~\label{def:corel}
    Given a category $\catC$ with pushouts and a costable factorisation system
   $(\Cepi,\Cmono)$, the category $\Corel{\catC}$ has the same objects as $\catC$.
   The arrows $X \to Y$ are equivalence classes of cospans $X \tr{f} N \tl{g} Y$ under
   the symmetric, transitive closure of the following relation: two cospans $X \tr{f} N
   \tl{g} Y$ and $X \tr{f'} N' \tl{g'} Y$ are related if there exists $N
   \tr{m}
   N'$ in $\Cmono$ such that
  \begin{equation}\label{eq:defcorelations}
    \begin{aligned}
    \xymatrix@R=5pt{
      & N   \ar[dd]^m  \\
      X \ar[ur]^{f} \ar[dr]_{f'}&& Y \ar[ul]_{g}\ar[dl]^{g'}\\
      & N'  
    }
  \end{aligned}
  \end{equation}
  commutes. This notion of equivalence respects composition of cospans, and so
  $\Corel{\catC}$ is indeed a category. We call the morphisms in this category
  {\bf corelations}. 
  
   Given a category $\catC$ with pullbacks and a stable factorisation
   system, we can dualise the above to define the category
   $\Rel{\catC}$ of {\bf relations}. 
  \end{defn}

(Co)stability is needed in order to ensure that composition of (co)relations is
associative, \emph{cf.}~\cite[\S 3.3]{Fon16}. For proofs it is convenient to
give an alternative description of (co)relations.

\begin{prop} \label{lemma:charCospanToCorel}
If $\catC$ has binary coproducts, corelations are in one-to-one
  correspondence with isomorphism classes of cospans such that the
  copairing $\copair{p,q} \maps X + Y \to N$ lies in $\Cepi$. 
  
If $\catC$ has binary products, relations are in one-to-one
  correspondence with isomorphism classes of spans such that the
  pairing $\pair{f,g} \maps N \to X \times Y$ lies in $\Cmono$. 
\end{prop}
\begin{proof}
We focus on relations, the proof for corelations being dual. It suffices to show
that two spans $\tl{f} \tr{g}$ and $\tl{f'}\tr{g'}$ represent the same relation
if and only if the $\Cmono$ parts of the factorisations of $\langle f,g\rangle$
and $\langle f',g' \rangle$ are equal.

For the backward direction, write factorisations $\langle f,g\rangle = e;m$ and
$\langle f',g'\rangle = e';m$, and note that $m = \langle m;p_1,m;p_2\rangle$,
where the $p_i$ are the canonical projections. Thus the following diagrams commute
  \[
    \xymatrix@R=5pt{
      &&  \ar[dll]_{f} \ar[ddd]^e \ar[drr]^{g} &&\\
      &&&& \\
      &  \ar[ul]^{p_1}&& \ar[ur]_{p_2} &\\
      &&   \ar[ul]^{m}  \ar[ur]_{m}&&
    } \qquad \qquad
        \xymatrix@R=5pt{
      &&  \ar[dll]_{f'} \ar[ddd]^{e'} \ar[drr]^{g'} &&\\
      &&&& \\
      &  \ar[ul]^{p_1}&&  \ar[ur]_{p_2} &\\
      &&   \ar[ul]^{m}  \ar[ur]_{m}&&
    }
  \]
Therefore $\tr{e \in \Cepi}$ and $\tr{e' \in \Cepi}$ witness that both $\tl{f} \tr{g}$ and
$\tl{f'}\tr{g'}$ are in the equivalence class of
$\tl{p_1}\tl{m}\tr{m}\tr{p_2}$ and so represent the same relation.

For the forward direction, note that if $\tl{f} \tr{g}$ and $\tl{f'}\tr{g'}$
represent the same relation, then there exists a sequence of spans $\tl{f_i}
\tr{g_i}$ in $\catC$ together with morphisms $\tr{e_i \in \Cepi}$,
$i=0, \dots n$, such that $f_1=f$, $g_1=g$, $f_n =f'$, $g_n= g'$, and for all $i
= 1,\dots ,n$ either (i) $e_i;f_i = f_{i-1}$ and $e_i;g_i=g_{i-1}$, or (ii) $f_i
=e_i;f_{i-1}$ and $g_i=e_i;g_{i-1}$. This implies either (i) $e_i;\langle f_i,g_i\rangle
=\langle f_{i-1},g_{i-1}\rangle$ or (ii) $\langle f_i,g_i\rangle
=e_i;\langle f_{i-1},g_{i-1}\rangle$. In either case, by the uniqueness of
factorisations, we see that the $\Cmono$ parts of $\langle f_i,g_i \rangle$ are
the same for all $i$.
\end{proof}

We call a span $\tl{f}\tr{g}$ {\bf jointly-in-$\Cmono$} if the pairing $\langle
f,g\rangle$ lies in $\Cmono$, and analogously call a cospan $\tr{f}\tl{g}$ {\bf
jointly-in-$\Cepi$} if the copairing $[f,g]$ lies in $\Cepi$. To each relation
there is thus, up to isomorphism, a canonical representation as a
jointly-in-$\Cmono$ span, and similarly to each corelation a jointly-in-$\Cepi$
cospan.

\begin{example} \label{ex.corels}
  Many examples of relations and corelations are already familiar.
  \begin{itemize}
    \item The category $\SET$ is bicomplete and has a bistable epi-mono
      factorisation system. Relations with respect to this factorisation system
      are simply the usual binary relations, while corelations from $X \to Y$ in
      $\SET$ are surjective functions $X+Y \to N$; thus their isomorphism
      classes---the arrows of $\Corel{\SET}$---are partitions, or equivalence relations on $X+Y$. 
    \item The category of vector spaces over a field $\field$ is abelian, and hence bicomplete with a
      bistable epi-mono factorisation system. The categories of
      relations and corelations are isomorphic: a morphism $X \to Y$ in these
      categories can be thought of as a linear relations, i.e. a subspace of $X \times Y$.  
       \item In any category $\catC$ the trivial morphism-isomorphism factorisation
      system $(\catC,\mathcal I_{\catC})$ is bistable. Relations with respect to $(\catC,\mathcal I_{\catC})$ are equivalence classes of isomorphisms $N
      \stackrel\sim\to X \times Y$, and hence there is a unique relation between
      any two objects. Corelations are just cospans.
    \item Dually, relations with respect to the isomorphism-morphism
      factorisation $(\mathcal I_{\catC}, \catC)$ are just spans, and there is
      a unique corelation between any two objects.  
        \end{itemize}
\end{example}

We now study the functorial interpretation of cospans and spans as
corelations. This discussion is instrumental in our universal construction for
corelations (Theorem \ref{thm:corelations}).

First, given two categories with the same collections of objects, we may speak
of {\bf identity-on-objects (ioo)} functors between them, i.e. functors that are
the identity map on objects. Four examples of such functors will become relevant in the next section:
\begin{equation}\label{eq:functorsCToSpanCospan}
\begin{aligned}
\catC \to \Cospan{\catC} \text{ maps } \tr{f} \text{ to } \tr{f}\tl{id} \qquad &
\qquad \phantom{\scriptstyle op}
\catC \to \Span{\catC} \text{ maps } \tr{f} \text{ to } \tl{id}\tr{f}\\
 \op{\catC} \to \Cospan{\catC} \text{ maps }  \tl{g } \text{ to } \tr{id}\tl{g}
 \qquad & \qquad
 \op{\catC} \to \Span{\catC} \text{ maps }  \tl{g } \text{ to } \tl{g}\tr{id}
 \end{aligned}
\end{equation}

We are now ready to discuss the canonical map from cospans to corelations. This
is simple: one just interprets a cospan representative as its corelation
equivalence class.

\begin{defn}~\label{def:cospantocorel}
Let $\catC$ be a category equipped with a costable factorisation system
  $(\Cepi,\Cmono)$. We define $\CospanToCorel \colon \Cospan{\catC} \to
  \Corel{\catC}$ as the ioo functor mapping the isomorphism class of cospans
  represented by $X\tr{f}N\tl{g}Y$ to the corelation represented by this cospan.
\end{defn}

It is straightforward to check that this is well-defined. Moreover, 
\begin{prop}\label{prop:CospanToCorelFull} $\CospanToCorel \colon \Cospan{\catC} \to \Corel{\catC}$ is full. \end{prop}
\begin{proof}
Let $a$ be a corelation. Then choosing some representative $X \to N \leftarrow
Y$ of $a$ gives a cospan whose $\CospanToCorel$-image is $a$.
\end{proof}

Mapping spans to corelations is more subtle. Given a span, we may obtain a
cospan by taking its pushout. When $\catC$ has pushouts and pullbacks, this
defines a function on morphisms $\Span{\catC} \to \Cospan{\catC}$. This function
is rarely, however, a functor: it may fail to preserve composition. To turn it
into a functor, two tweaks are needed: first, we restrict to a subcategory
$\Span\subc$ of $\Span{\catC}$, for some carefully chosen subcategory $\subc
\subseteq \catC$, and second, we take the jointly-in-$\Cepi$ part of the
pushout. We call the resulting functor $\SpanToCorel$, as it takes the
\emph{pushout} and then \emph{projects}.

How do we choose $\subc$? Given a cospan $X \to A \leftarrow Y$, we may take its
pullback to obtain a span $X \leftarrow P \to Y$, and then pushout this span in
$\catC$ to obtain a cospan $X \to Q \leftarrow Y$. This gives a diagram
\[
  \xymatrix@R=10pt@C=40pt{
    & X \ar[dr] \ar@/^1em/[drr] \\
    P \ar[ur] \ar[dr] && Q \ar@{-->}[r]^f & A \\
    & Y \ar[ur] \ar@/_1em/[urr]
  }
\]
where the map $f$ exists and is unique by the universal property of the pushout.
We want this map $f$ to lie in $\Cmono$: by Definition~\ref{def:corel}, this implies that $X \to A \leftarrow
Y$ and $X \to Q \leftarrow Y$ represent the same corelation. This condition is
reminiscent of that introduced by Meisen in her work on so-called categories of
pullback spans \cite{Mei74}.

Note that this pullback and pushout take place in $\catC$. We nonetheless
ask $\subc$ to be closed under pullback, so spans $\tl{f \in \subc}\tr{g\in \subc}$ do indeed form a subcategory $\Span\subc$ of $\Span\catC$.\footnote{Calling this subcategory $\Span\subc$ is a slight abuse of notation: it may be
the case that $\subc$ itself has pullbacks, and we have not proved that these
agree with pullbacks in $\catC$. Nonetheless, this conflict does not cause
trouble in any of our examples below, and we stick to this convention for
notational simplicity.}

\begin{prop}~\label{prop:spantocorel}
  Let $\catC$ be a category equipped with a costable factorisation system
  $(\Cepi,\Cmono)$.  Let $\subc$ be a subcategory of $\catC$ containing all
  isomorphisms and stable under pullback.  Further suppose that the canonical
  map given by the pushout of the pullback of a cospan in $\subc$ lies in
  $\Cmono$. Then mapping a span in $\subc$ to the jointly-in-$\Cepi$ part of its
  pushout cospan defines an ioo functor $\Pi\maps \Span{\subc} \to \Corel{\catC}$.
\end{prop}
\begin{proof}
  Recall that $\Span{\subc}$ is generated by morphisms of the form
  $\tl{id}\tr{f \in \subc}$ and $\tl{f \in \subc}\tr{id}$. It is thus enough to
  show $\SpanToCorel$ preserves composition on arrows of these two types. There
  exist four cases: (i) $\tl{id}\tr{f}\tl{id}\tr{g}$, (ii)
  $\tl{f}\tr{id}\tl{g}\tr{id}$, (iii) $\tl{f}\tr{id}\tl{id}\tr{g}$, and (iv)
  $\tl{id}\tr{f}\tl{g}\tr{id}$. The first three cases are straightforward to
  prove, and in fact hold when mapping $\Span{\catC} \to \Cospan{\catC}$. It is
  the case (iv) that needs our restriction to $\Span\subc$. There $\SpanToCorel(\tl{id}\tr{f})\SpanToCorel(\tl{g}\tr{id})$
  is represented by the cospan $\tr{f}\tl{g}$, while
  $\SpanToCorel(\tl{id}\tr{f}\tl{g}\tr{id})$ is the represented by the
  pushout $\tr{p}\tl{q}$ of the pullback of $\tr{f}\tl{g}$. But by hypothesis,
  there exists a unique $\tr{m \in \Cmono}$ making the following diagram commute.
  \[
    \xymatrix@C=30pt@R=10pt{
      &  \ar[dr]^p \ar@/^1em/[drr]^f \\
      &&  \ar[r]^m &  \\
      & \ar[ur]_q \ar@/_1em/[urr]_g
    }
  \]
This implies that $\tr{p}\tl{q}$ and $\tr{f}\tl{g}$ represent the
  same corelation, and so $\SpanToCorel$ is functorial.
\end{proof}

For example, if the category $\Cmono$ has pullbacks and these coincide with pullbacks in
$\catC$, then we can take $\subc=\Cmono$. If $\catC$ is abelian, we can take $\subc=\catC$.

\section{Main theorem: a universal property for (co)relations}\label{sec:theorem}

This section states our main result and some consequences.

\subsection{The main theorem}
We first fix our ingredients. 
\begin{ass}\label{ass:thcorelations}
  Let $\catC$ be a category with
  \begin{itemize}
    \item pushouts and pullbacks;
  \item a costable factorisation system $\fact{\Cepi}{\Cmono}$ with $\Cmono$ a subcategory of the
  monos in $\catC$;% and stable under pushout;
  \item a subcategory $\subc$ of $\catC$ containing $\Cmono$, stable under pullback, and
  such that the canonical map given by the pushout of the pullback of a
  cospan in $\subc$ lies in $\Cmono$.
  \end{itemize}
\end{ass}

Building on the results of Section~\ref{sec:corelations}, the second requirement
above allows us to form a category $\Corel{\catC}$ of corelations, whereas the
third yields a functor $\SpanToCorel \colon \Span{\subc} \to \Corel{\catC}$.  

We shall also use the functor $\CospanToCorel \colon \Cospan{\catC} \to
\Corel{\catC}$ (Definition~\ref{def:cospantocorel}), and a category
$\subc\+{\subc}\op{\subc}$. The objects of this latter category are those of
$\subc$ and the morphisms $X \to Y$ are so-called `zigzags' $X\tr{f}\tl{g}\tr{h}
\dots \tl{k} Y$ in $\subc$, where $f$, $g$, $\dots$, $k$ are non-identity
morphisms in $\subc$.  Composition in $\subc\+{\subc}\op{\subc}$ is given by
concatenation of zigzags, where we use composition in $\subc$ to compose any two
consecutive morphisms pointing in the same direction, and remove any resulting
identity morphisms in $\subc$. The identity morphism in
$\subc\+{\subc}\op{\subc}$ is the empty zigzag consisting of no morphisms. Note
that there are ioo functors from $\subc\+{\subc}\op{\subc}$ to $\Cospan{\catC}$
and to $\Span{\catC}$, defined on morphisms by taking colimits and by taking
limits, respectively, of zigzags---equivalently, they are defined by pointwise
application of the functors in~\eqref{eq:functorsCToSpanCospan}.\footnote{More
abstractly, one can see $\subc\+{\subc}\op{\subc}$ as the pushout of $\subc$ and
$\op{\subc}$ over the respective inclusions of $\lvert\subc\rvert$, the discrete
category on the objects of $\subc$. The functors $\Span{\subc} \tl{}
\subc\+{\subc}\op{\subc} \tr{} \Cospan{\catC}$ are then those given by the
universal property with respect to (suitable restrictions of) the functors in
\eqref{eq:functorsCToSpanCospan}.} 

We make all these components interact in our main theorem.

\begin{thm}
 \label{thm:corelations}
   With $\catC$ and $\subc$ as in Assumption \ref{ass:thcorelations}, the
   following is a pushout in $\CAT$:
\begin{equation}\tag{$\star$}
\label{eq:pushoutCorel}
\begin{aligned}
    \xymatrix@C=40pt{
      {\subc \+{\subc} \op{\subc}} \ar[r] \ar[d] & {\Spanc{\subc}}
      \ar[d]^{\SpanToCorel} \\
    {\Cospan{\catC}} \ar[r]_-{\CospanToCorel} & {\Corel{\catC}} \pullbackcorner
    }
  \end{aligned}
\end{equation}
\end{thm}
We leave a complete proof of this theorem to Appendix~\ref{sec:proof}. In a
nutshell, the key point is that, in light of~\eqref{eq:defcorelations},
$\Corel{\catC}$ differs from $\Cospan{\catC}$ precisely because it has the extra
equations $\tr{m}\tl{m} = \tr{\id}\tl{\id}$, with $\tr{m} \in \Cmono$. But these
equations arise by pullback squares in $\subc$, and so are equations of zigzags
in $\Span{\subc}$ (\emph{cf.}\ \eqref{eq:quotientPb} below). Moreover, the
remaining equations of $\Span{\subc}$ can be generated by using these together
with a subset of the equations of $\Cospan{\catC}$. Hence adding the equations
of $\Span{\subc}$ to those of $\Cospan{\catC}$ gives precisely $\Corel{\catC}$,
and we have a pushout square.

\medskip

We now discuss some observations, consequences, and examples.

\begin{rmk}\label{rmk:Msuffices}
If any such $\subc$ exists, then we may always take $\subc=\Cmono$ and
the theorem holds. We record the above, more general, theorem as it explains
preliminary results in this direction already in the literature; see the abelian
case and examples for details.
\end{rmk}

\begin{example}
To briefly fix some examples in mind, we remark that
Assumption~\ref{ass:thcorelations} is satisfied by the (epi, mono) factorisation
systems in the categories of (finite) sets and functions, vector spaces and
linear transformations, and any presheaf category. Here, following
Remark~\ref{rmk:Msuffices}, we simply take the subcategory $\subc$ to be the
subcategory of monos, although other choices also suffice. It can also be shown
that, given any pair of categories equipped with data satisfying
Assumption~\ref{ass:thcorelations}, new examples may be obtained by taking their
product and coproduct. Theorem~\ref{thm:corelations} thus applies to all these
cases. We will explore some of these examples, and many others, in detail in
Section~\ref{sec:examples}.
\end{example}

We may easily formulate the dual version of the theorem, which yields a characterisation for relations. It is based on a dual version of Assumption \ref{ass:thcorelations}.

\begin{ass}\label{ass:threlations}
  Let $\catC$ be a category with
  \begin{itemize}
    \item pushouts and pullbacks;
  \item a stable factorisation system $\fact{\Cepi}{\Cmono}$ with $\Cepi$ a subcategory of the
  epis in $\catC$; %and stable under pullback;
  \item a subcategory $\subc$ of $\catC$ containing $\Cepi$, stable under
pushout, and such that the canonical map given by the pullback of the pushout of
a span in $\subc$ lies in $\Cepi$.
  \end{itemize}
\end{ass}

\begin{cor}[Dual case] \label{cor.dual}
  With $\catC$ and
  $\subc$ as in Assumption \ref{ass:threlations}, the following is a pushout in $\CAT$.
  \begin{equation}\tag{$\triangle$}\label{eq:pushoutRel}
    \begin{aligned}
      \xymatrix@C=40pt{
	{\subc \+{\subc} \op{\subc}} \ar[r] \ar[d] & {\Cospan{\subc}}
	\ar[d]^{} \\
	{\Span{\catC}} \ar[r]_-{} & {\Rel{\catC}} \pullbackcorner
      }
    \end{aligned}
  \end{equation}
\end{cor}
\begin{proof}
  This corollary is obtained by noting that, given a stable factorisation system
  $(\Cepi,\Cmono)$ in $\catC$, with $\Cepi$ a subcategory of the epis, we have a
  costable factorisation system $(\op{\Cmono},\op{\Cepi})$ in $\op{\catC}$, with
  $\op\Cepi$ a subcategory of the monos. Proposition \ref{prop:spantocorel} then
  gives a functor $\Cospan{\subc}= \Span{\op\subc} \longrightarrow \Rel{\catC} =
  \Corel{\op\catC}$. Noting also that $\subc \+{\subc} \op{\subc} = \op\subc
  \+{\subc} \op{(\op\subc)}$ and $\Span{\catC} = \Cospan{\op\catC}$, we can
  hence apply Theorem~\ref{thm:corelations}.
\end{proof}

\begin{rmk}
  In Theorem~\ref{thm:corelations}, the diagram $\Cospan{\catC} \leftarrow \subc
  \+{\subc} \subc \to \Span{\subc}$ `knows' only about $\Cmono$, not the
  factorisation system $(\Cepi,\Cmono)$. One might then wonder how the pushout
  could be known to be $\Corel{\catC}$, since this is the category of
  jointly-in-$\Cepi$ cospans. This is enough, however, since if $\Cmono$ is part
  of a factorisation system, then the factorisation system is unique.

  Indeed, suppose we have $\Cepi$, $\Cepi'$ such that both $(\Cepi,\Cmono)$ and
  $(\Cepi',\Cmono)$ are factorisation systems. Take $e \in \Cepi$. Then the
  factorisation system $(\Cepi',\Cmono)$ gives a factorisation $e = e';m_1$,
  while $(\Cepi,\Cmono)$ gives a factorisation $e' = e_2;m_2$.
  By substitution, we have $e = e_2;m_2;m_1$. By uniqueness of factorisation, we can then assume without loss of
  generality that $e = e_2$ and $m_2;m_1 = id$. Next, using $e=e_2$ and
  substitution in $e' = e_2;m_2$, we similarly arrive at $m_1;m_2 = id$. Thus $m_1$ is
  an isomorphism, and hence lies in $\Cepi'$. This implies that $e = e';
  m_1 \in \Cepi'$, and hence $\Cepi \subseteq \Cepi'$. We may similarly show
  that $\Cepi' \subseteq \Cepi$, and hence that the two categories are equal.
\end{rmk}

% \begin{rmk} Both the characterisation for corelations and for relations hold in a more general version for categories that are not props. We chose to present the result for props as the presentation is simpler and all our case study are props.\end{rmk}

The next corollary is instrumental in giving categories of (co)relations a presentation by generators and equations.

\begin{cor}\label{cor:presentation} Suppose $\subc$ and $\catC$ are as in Assumption \ref{ass:thcorelations}. Then $\Corel{\catC}$ is freely generated by the objects of $\catC$ and arrows $\tr{f}$, $\tl{g}$ of $\catC$ quotiented by 
\begin{align}
 \tr{f \in \subc}\tl{g \in \subc} \ = \ \tl{p \in \subc}\tr{q \in \subc} && \text{whenever }\tl{p}\tr{q} \text{ pulls back } \tr{f}\tl{g} \label{eq:quotientPb} \\
 \tl{f \in \catC}\tr{g \in \catC} \ = \ \tr{p \in \catC}\tl{q \in \catC} && \text{whenever }\tr{p}\tl{q} \text{ pushes out } \tl{f}\tr{g}. \label{eq:quotientPo}
\end{align}
Equivalently, $\Corel{\catC}$ is the quotient of $\Cospan{\catC}$ by~\eqref{eq:quotientPb}. A dual statement holds for $\Rel{\catC}$.
\end{cor}

Note that, in light of Remark \ref{rmk:Msuffices}, one may also replace \eqref{eq:quotientPb} by the subset of axioms
\begin{align*}
 \tr{f \in \Cmono}\tl{g \in \Cmono} \ = \ \tl{p \in \Cmono}\tr{q \in \Cmono} && \text{whenever }\tl{p}\tr{q} \text{ pulls back } \tr{f}\tl{g}.
\end{align*}
As $\Cmono \subseteq \subc$ by Assumption~\ref{ass:thcorelations}, this may give a smaller presentation.

Corollary \ref{cor:presentation} provides a modular recipe for axiomatising categories of
(co)relations: a presentation for $\Corel{\catC}$ can be obtained by ``merging'' equational theories arising from $\subc$-pullbacks and $\catC$-pushouts, provided respectively by sets of equations \eqref{eq:quotientPb} and \eqref{eq:quotientPo}. t is worth noticing that equations \eqref{eq:quotientPb}-\eqref{eq:quotientPo} enjoy an elegant description in terms of \emph{distributive laws} of categories, see \cite{RosebrRWood_fact,Lack2004a,ZanasiThesis}. In practice, the reduction of Corollary \ref{cor:presentation} is particularly useful when $\catC$ is ``nice enough'' to allow for a \emph{finitary} axiomatisation of equations \eqref{eq:quotientPb}-\eqref{eq:quotientPo}--- two such examples are the category of finite sets, in which pullbacks decompose according to the laws of extensive categories \cite{Lack2004a}, and the category of finite-dimensional vector spaces, where pullbacks/pushouts can be computed by kernels/cokernels of matrices \cite{interactinghopf}. We shall see numerous applications of the recipe of Corollary \ref{cor:presentation} in the examples of Section \ref{sec:examples}
below. 

In the remainder of this section, we discuss the consequences and
generalisations of the main theorem when various additional structure is
present, including the cases of abelian categories, monoidal categories, and
props.

\subsection{The abelian case} 
As a notable instance of Theorem~\ref{thm:corelations}, we can specialise to the
case of abelian categories and their epi-mono factorisation system. In this case
we can simply pick $\subc$ to be $\catC$ itself.

\begin{cor}[Abelian case]
 \label{thm:pushoutAbelian}
  Let $\catC$ be an abelian category. Then the following is a pushout square in
  $\CAT$:
\[
\begin{aligned}
    \xymatrix@C=40pt{
    {\catC \+{\catC} \op{\catC}} \ar[r]
    \ar[d] & {\Span{\catC}}
    \ar[d]^{\SpanToCorel} \\
    {\Cospan{\catC}} \ar[r]_-{\CospanToCorel} & {\Corel{\catC}} \cong \Rel{\catC} \pullbackcorner
    }
  \end{aligned}
\]
  where we take (co)relations with respect the epi-mono factorisation system. 
\end{cor}
\begin{proof}
  As $\catC$ is abelian, it is finitely bicomplete and has a bistable
  factorisation system given by epis and monos.  Furthermore, we need not
  restrict our spans to some subcategory $\subc$: in an abelian category the
  pullback of a cospan $X\tr{f} A\tl{g}Y$ can be computed via the kernel of the
  joint map $X \oplus Y\tr{[f,-g]}A$, and similarly pushouts can be computed via
  cokernel, whence the canonical map from the pushout of the pullback of a given
  cospan to itself is always mono, being the inclusion of the image of the joint
  map into the apex.  Similarly, the map from a span to the pullback of its
  pushout is simply the joint map with codomain restricted to its image, and
  hence always epi.  Thus $\catC$ meets both Assumptions \ref{ass:thcorelations}
  and \ref{ass:threlations} with $\subc = \catC$.  Then the category of
  corelations is the pushout of the span $\Cospan{\catC} \leftarrow \catC
  \+{\catC} \op{\catC} \to \Span{\catC}$. But by the dual theorem (Corollary
  \ref{cor.dual}), the pushout of this span is also the category of relations.
  Thus the two categories are isomorphic.  Explicitly, the isomorphism is given
  by taking a corelation to the jointly mono part of its pullback span, and
  taking a relation to the jointly epi part of its pushout cospan.
\end{proof}

\subsection{The monoidal case}\label{sec:monoidal}
Categories of (co)spans and (co)relations are often monoidal categories with,
for example, the monoidal product coming from (co)products. Our theorem extends
to this monoidal case. 

For monoidal structure on $\catC$ to extend to the categories of (co)spans and
(co)relations, it is crucial that the monoidal product respects the ambient
structure. Let $(\catC,\tns)$ be a symmetric monoidal category with pushouts,
and let $(\subc,\tns)$ be a sub-symmetric monoidal category. We say that {\bf
the monoidal product preserves pushouts} in $\subc$ if, for all spans $N
\leftarrow Y \to M$ and $N' \leftarrow Y' \to M'$ in $\subc$, we have an
isomorphism
  \[
    (N \tns N') +_{Y \tns Y'} (M \tns M') \cong (N +_Y M) \tns (N' +_{Y'} M').
  \]
  Note that this pushout is taken in $\catC$. This condition holds, for example,
  whenever $\catC$ is monoidally closed. We say the monoidal product preserves
  pullbacks if the analogous condition holds for pullbacks.

  Furthermore, we say that a subcategory $\subc$ is {\bf closed under $\tns$}
  if, given morphisms $f,g$ in $\subc$, the morphism $f\oplus g$ is also in
  $\subc$.

\begin{prop}\label{thm:corelationsSMCs}
  Let $\catC$ and $\subc$ be symmetric monoidal categories satisfying Assumption
  \ref{ass:thcorelations}. Suppose that the monoidal product of $\catC$ preserves
  pushouts in $\catC$ and pullbacks in $\subc$, and that $\Cmono$ and $\subc$ are
  closed under the monoidal product. Then the following is a pushout square in the
  category of symmetric monoidal categories and (lax, strong, strict) monoidal
  functors.
  \[
    \label{eq:pushoutCorelSMCs}
    \begin{aligned}
    \xymatrix@C=40pt{
	  {\subc +_{\lvert \subc\rvert} \op{\subc}} \ar[r] \ar[d] & {\Spanc{\subc}}
	  \ar[d]^{\SpanToCorel} \\
	  {\Cospan{\catC}} \ar[r]_-{\CospanToCorel} & {\Corel{\catC}} \pullbackcorner.
    }
    \end{aligned}
  \]
\end{prop}

The proof of this result relies on an analysis of how monoidal structure carries through the various steps of our pushout construction. The details can be found in Appendix~\ref{app:spansPROP}.

\subsection{The case of props} As mentioned in the introduction, the motivating
examples for our construction are categories providing a semantic interpretation
for circuit diagrams. These are typically props (\textbf{pro}duct and
\textbf{p}ermutation categories~\cite{MacLane1965}): it is thus useful to phrase
our construction in this setting. This is merely a specialisation of the above
result.

Recall that a prop is a symmetric monoidal category with objects the
natural numbers, in which $n \tns m = n+m$. Props form a category $\PROP$ with
morphisms the ioo strict symmetric monoidal functors. A simplification to Theorem \ref{thm:corelations} is that the
coproduct $\catC + \catC'$ in $\PROP$ is computed as $\catC \+{\catC} \catC'$ in
$\CAT$, because the set of objects is fixed for any prop.

\begin{thm}\label{thm:corelationsPROPs}
  Let $\catC$ and $\subc$ be props satisfying Assumption
  \ref{ass:thcorelations}. Suppose that the monoidal product of $\catC$
  preserves pushouts in $\catC$ and pullbacks in $\subc$, and that $\Cmono$
  is closed under the monoidal product. Then we have a pushout square in $\PROP$
\begin{equation}
\label{eq:pushoutCorelProps} \tag{$\circ$}
\begin{aligned}
    \xymatrix@C=40pt{
      {\subc + \op{\subc}} \ar[r] \ar[d] & {\Spanc{\subc}}
      \ar[d]^{\SpanToCorel} \\
    {\Cospan{\catC}} \ar[r]_-{\CospanToCorel} & {\Corel{\catC}} \pullbackcorner.
    }
  \end{aligned}
\end{equation}
\end{thm}
\begin{proof}
  Recalling that props are certain structured strict symmetric monoidal
  categories and prop morphisms are structured strict monoidal functors, the
  theorem is an immediate consequence of Proposition~\ref{thm:corelationsSMCs}.
\end{proof}

We also state the prop version of the abelian case. An abelian prop is just a
prop which is also an abelian category and where the monoidal product is the
biproduct.

\begin{cor}
 \label{thm:corelationsAbPROP}
  Suppose that $\catC$ is an abelian prop. The following is a pushout in $\PROP$.
\[
\begin{aligned}
    \xymatrix@C=40pt{
    {\catC + \op{\catC}} \ar[r]
    \ar[d] & {\Span{\catC}}
    \ar[d]^{\SpanToCorel} \\
    {\Cospan{\catC}} \ar[r]_-{\CospanToCorel} & {\Corel{\catC}} \cong \Rel{\catC} \pullbackcorner
    }
  \end{aligned}
\]
\end{cor}
\begin{proof}
  Note that in an abelian prop the biproduct, being both a product and a
  coproduct, preserves both pushouts and pullbacks, and that the monos are
  closed under the biproduct. Thus we can apply
  Theorem \ref{thm:corelationsPROPs}.
\end{proof}

The proof of Theorem \ref{thm:corelationsPROPs} and of Corollary can be found in Appendix \ref{app:spansPROP}.

\subsection{(Co)algebras over a (co)monad} \label{prop.algebras}
%Recall that a morphism is a regular epi if it can be written as a coequalizer.
It is of interest to identify under which conditions our universal construction of $\Rel{\catC}$ lifts to relations over $\catC^T$, the category of algebras 
of some monad $T \colon \catC \to \catC$.

\begin{prop} \label{prop:algebrasmonad}
  Let $\catC$ be a complete and cocomplete regular category that obeys
  Assumption~\ref{ass:threlations} with respect to its (regular epi, mono)
  factorisation system. Further assume that every regular epi splits.  Next, let
  $T$ be a monad on $\catC$ such that $T$ and $T^2$ preserve pushouts of regular
  epis.  Then the category $\catC^T$ of algebras over $T$ is a complete and
  cocomplete regular category that obeys Assumption~\ref{ass:threlations} with
  respect to its (regular epi, mono) factorisation system. 
\end{prop}

Note that if $\catC$ is a complete and cocomplete category, then regularity is
just the statement that regular epis are stable under pullback, and
Assumption~\ref{ass:threlations} is just the statement that the canonical map
given by the pullback of the pushout of a span of regular epis is again a
regular epi.

\begin{proof}
The completeness, cocompleteness, and regularity
  of $\catC^T$ under the above conditions is a standard result, see e.g. \cite[Th. 4.3.5]{Bor94}. In brief, the completeness of $\catC^T$ is
  easily inherited from the completeness of $\catC$. The cocompleteness requires
  the regular epis to split and the existence of products.  This allows explicit
  construction of coequalizers in $\catC^T$ and, using the existence of
  coproducts in $\catC$, this immediately implies the cocompleteness of
  $\catC^T$ (see \cite[Prop.  4.3.4]{Bor94}). The stability of regular epis
  under pullback in $\catC^T$ follows from the fact that the forgetful functor
  $U\colon \catC^T \to \catC$ preserves and reflects both pullbacks and regular
  epis. Thus, since regular epis are stable under pullback in $\catC$, the same
  holds true in $\catC^T$.

  It then remains to show that, in $\catC^T$, the canonical map given by the
  pullback of the pushout of a span of regular epis is again a regular epi.
  Since $T$ preserves pushouts of regular epis, such a pushout in $\catC^T$ may
  simply be obtained by computing the pushout $P$ of the underlying maps in
  $\catC$, and using the universal property of the $T$-image of this pushout to
  get a map $c\colon TP \to P$.  Since $T^2$ also preserves pushouts of regular
  epis, it can be shown that this map $c$ is in fact an algebra, and the pushout
  in $\catC^T$ (see \cite[Prop.  4.3.2]{Bor94}). On the other hand, as limits in
  a category of algebras, pullbacks in $\catC^T$ can easily be computed via
  pullbacks in $\catC$.  Thus the canonical map in question is simply the
  canonical map in $\catC$ lifted to a map of algebras, and hence is again a
  regular epi. This proves the claim.
\end{proof}

We may dualise Proposition \ref{prop:algebrasmonad} to the case of corelations
in the category of coalgebras over a comonad; in this case we assume our
category is coregular, and that it is the regular monos that split in $\catC$
and whose pullbacks are preserved by the comonad $T$ and by $T^2$.

Note also that the assumptions above can be varied or weakened to still obtain
that the category of (regular epi, mono) relations in a category of algebras can
be constructed as a pushout. For example, it is straightforward to relax
cocompleteness to, say, finite cocompleteness, or to put conditions on the
properties of $\catC^T$ instead of on $\catC$ and $T$.\footnote{This discussion amends a mistake in the conference version of this paper
\cite{FZ-calco17}. The conditions stated therein were too weak for the above proof to work. More precisely, the problem was that pushouts in $\catC^T$ were not guaranteed to sufficiently
resemble pushouts in $\catC$ for the pushout--pullback condition to hold.}

\section{Examples}\label{sec:examples}

\subsection{Equivalence Relations} \label{ex:er}

Our first example concerns the construction of equivalence relations starting
from injective functions. For $n \in \mathbb{N}$, write $\ord{n}$ for the set
$\{0,1,\dots,n-1\}$, and $\uplus$ for the disjoint union of sets. We fix a prop
$\ER$ whose arrows $n \to m$ are the equivalence relations on
$\ord{n}\uplus\ord{m}$. For composition $e_1 \poi e_2 \maps n \to m$ of
equivalence relations $e_1 \maps n \to z$ and $e_2 \maps z \to m$, one first
defines an equivalence relation on $\ord{n}\uplus\ord{z}\uplus\ord{m}$ by gluing
together equivalence classes of $e_1$ and $e_2$ along common witnesses in
$\ord{z}$, then obtains $e_1 \poi e_2$ by restricting to elements of
$\ord{n}\uplus\ord{m}$. 
 
Equivalence relations are equivalently described as corelations of functions. For this, let $\F$ be the prop whose arrows $n \to m$ are functions from $\ord{n}$ to $\ord{m}$. $\F$ has the usual factorisation system $(\Surj,\Inj)$ given by epi-mono factorisation, where $\Surj$ and $\Inj$ are the sub-props of surjective and of injective functions respectively. Given these data, one can check that $\ER$ is isomorphic to $\Corel{\F}$, the prop of corelations on $\F$.

We are now in position to apply our construction of
Theorem~\ref{thm:corelationsPROPs}. First, we verify
Assumption~\ref{ass:thcorelations} with $\catC$ instantiated as $\F$ and $\subc$
as $\Inj$. The only point requiring some work is the third, which goes as
follows: given a cospan of monos $X \to P \leftarrow Y$, consider $X$, $Y$ as
subsets of $P$. Then the pullback-pushout diagram looks like
\[
  \xymatrix@R=5pt@C=20pt{
    & X \ar[dr] \ar@/^/[drr] \\
    X\cap Y \ar[ur] \ar[dr] && X\cup Y \ar@{-->}[r] & P \\
    & Y \ar[ur] \ar@/_/[urr]
  }
\]
and $X\cup Y \to P$ is the inclusion map, hence a morphism in $\Inj$. Therefore, we
can construct the pushout diagram~(\ref{eq:pushoutCorelProps}) as follows:
\begin{equation}\
\label{eq:pushoutER}
\begin{aligned}
    \xymatrix@R=18pt{
    {\Inj + \op{\Inj}} \ar[r]
    \ar[d]& {\Span{\Inj}}
    \ar[d] \\
    {\Cospan{\F}} \ar[r] & {\ER} \pullbackcorner
    }
  \end{aligned}
\end{equation}

This modular reconstruction easily yields a presentation by generators and relations for (the arrows of) $\ER$. Following the recipe of Corollary \ref{cor:presentation}, $\ER$ is the quotient of $\Cospan{\F}$ by all the equations generated by pullbacks in $\Inj$, as in~\eqref{eq:quotientPb}. Now, recall that $\Inj$ is presented (in string diagram notation~\cite{Selinger2009}) by the generator $\Bunit \colon 0 \to 1$, and no equations. Thus, in order to present all the equations of shape~\eqref{eq:quotientPb} it suffices to consider one pullback square in $\Inj$:
\begin{equation}\label{eq:spaninj}
\vcenter{\xymatrix@=10pt{
& 1 & \\
0\ar[ur]^{\Bunit} && 0 \ar[ul]_{\Bunit} \\
& \ar[ul]^{\EmptyDiag} 0 \ar[ur]_{\EmptyDiag} &
}} \quad \text{, yielding the equation } \Bunit \poi \Bcounit \ = \ \EmptyDiag \poi \EmptyDiag.
\end{equation}
On the other hand, we know $\Cospan{\F}$ is presented by the theory of special
commutative Frobenius monoids (also termed separable Frobenius algebras),
see~\cite{Lack2004a}. Therefore $\ER$ is presented by the generators and
equations of special commutative Frobenius monoids, with the addition
of~\eqref{eq:spaninj}. This is known as the theory of extraspecial commutative
Frobenius monoids~\cite{CF}. This result also appears in~\cite{Zanasi16}, in both cases
without the realisation that it stems from a more general construction.

\subsection{Partial Equivalence Relations} \label{ex:per}

Partial equivalence relations (PERs) are common structures in program semantics, which date back to the seminal work of Scott~\cite{Scott1975-datatypeslattices} and recently revamped in the study of quantum computations (e.g., \cite{Jacobs_quantumPER,HasuoH11_GOIQuantumPER}). Our approach yields a characterisation for the prop $\PER$ whose arrows $n \to m$ are PERs on $\ord{n}\uplus\ord{m}$, with composition as in $\ER$. The ingredients of the construction generalise Example~\ref{ex:er} from total to partial maps. Instead of $\F$ one starts with $\PF$, the prop of partial functions, which has a factorisation system involving the sub-prop of partial surjections and the sub-prop of injections. The resulting prop of $\PF$-corelations is isomorphic to $\PER$. Theorem~\ref{thm:corelationsPROPs} yields the following pushout
\begin{equation}\
\label{eq:pushoutPER}
\begin{aligned}
    \xymatrix@R=18pt{
    {\Inj + \op{\Inj}} \ar[r]
    \ar[d]& {\Span{\Inj}}
    \ar[d] \\
    {\Cospan{\PF}} \ar[r] & {\PER} \pullbackcorner .
    }
  \end{aligned}
\end{equation}
As in Example~\ref{ex:er}, following Corollary \ref{cor:presentation},~\eqref{eq:pushoutPER} reduces the task of axiomatising $\PER$ to the one of axiomatising $\Cospan{\PF}$ and adding the single equation~\eqref{eq:spaninj} from $\Span{\Inj}$. $\Cospan{\PF}$ is presented by ``partial'' special commutative Frobenius monoids, studied in~\cite{Zanasi16}.

\subsection{Non-Example: the Terminal Prop}\label{ex:triv}

 It is instructive to see a non-example, to show that the assumptions on $\subc$ are not redundant. One may want consider an obvious variation of~\eqref{eq:pushoutER}, where instead of $\Inj$ one takes the whole $\F$ as $\subc$. However, with this tweak the construction collapses: the pushout is the terminal prop $\Triv$ with exactly one arrow between any two objects. 
\[
\begin{aligned}
    \xymatrix@R=18pt{
    {\F + \op{\F}} \ar[r]
    \ar[d] & {\Span{\F}} \ar[d]^{} \\
    {\Cospan{\F}} \ar[r]_-{} & {\Triv} \pullbackcorner
    }
  \end{aligned}
\]
This phenomenon was noted before (\cite{interactinghopf}, see also~\cite[Th.
5.6]{HeunenVicaryBook}), however without an understanding of its relationship
with other (non-collapsing) instances of the same construction.
Theorem~\ref{thm:corelationsPROPs} explains why this case fails where others
succeed: the problem lies in the choice of $\F$ as the subcategory $\subc$.
Indeed, the canonical map given by the pullback of the pushout of any span in
$\subc = \F$ does not necessarily lie in $\Inj$, i.e. it may be not injective. An
example is given by the cospan $0 \to 1 \leftarrow 2$, with the canonical map
from the pushout of the pullback cospan the non-injective map $2 \to 1$:
\[
  \xymatrix@R=5pt@C=35pt{
    & 2 \ar[dr] \ar@/^/[drr] \\
    0 \ar[ur] \ar[dr] && 2 \ar@{-->}[r] & 1. \\
    & 0 \ar[ur] \ar@/_/[urr]
  }
\]

\subsection{Non-Example: Spans of Functions} \label{ex:plainrelations}

Another interesting non-example is the attempt of constructing $\Rel{\F}$, the prop of relations between finite sets. Inspired by the case of corelations (Section \ref{ex:er}), one may think of using again $\Inj$ as subcategory $\subc$ and apply the recipe~\eqref{eq:pushoutRel} from Corollary \ref{cor.dual}. However, because $\Cospan{\Inj}$ and $\Inj + \op{\Inj}$ are isomorphic props (the distributive law forming $\Cospan{\Inj}$ does not yield any additional equations), the pushout diagrams just yields spans again.
\begin{equation*}
\begin{aligned}
    \xymatrix@R=18pt{
    {\Inj + \op{\Inj}} \ar[r]
    \ar[d] & {\Span{\F}} \ar[d]^{} \\
    {\Cospan{\Inj}} \ar[r]_-{} & {\Span{\F}} \pullbackcorner
    } 
  \end{aligned}
\end{equation*}
The reason why the construction~\eqref{eq:pushoutRel} does not apply is because
the third item of Assumption~\ref{ass:threlations} fails. The following is a
counterexample, where we start with a span $2 \tl{} 3 \tr{} 2$ of epis, take the
pushout $2 \tr{} 1 \tl{} 2$, then the pullback $2 \tr{} 4 \tl{} 2$, and then
observe that the canonical morphism $3 \to 4$ cannot be epi. As before, the
functions involved are specified in string diagrammatic notation. 
\qquad \qquad
    \[
  \xymatrix@R=8pt@C=55pt{
   & \ar[dl]_{\cgr{Bmult.pdf}} 2 &  \\
  1  &&  4 \ar[ul]_{\cgr{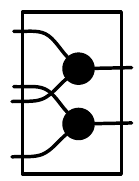}} \ar[dl]^{\cgr{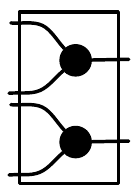}}  &
  \ar@/_3em/[ull]_{\cgr{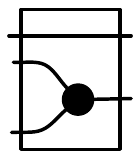}} \ar@/^3em/[dll]^{\cgr{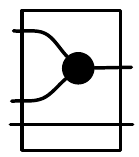}}
  \ar@{-->}[l] 3. \\
  & \ar[ul]^{\cgr{Bmult.pdf}} 2 & 
  }
\]

\subsection{Isomorphisms} \label{ex:iso}

In a degenerate way, our construction can actually be applied to an arbitrary category $\catC$ with pushouts and pullbacks. Indeed, we noted in Section \ref{sec:corelations}
that every category $\catC$ has a factorisation system $(\catC,\Ciso)$, where
again $\Ciso$ is the subcategory containing exactly the isomorphisms. Note that
$\Ciso$ is always a subcategory of the category of monos in $\catC$: we may then
also take $\subc=\Ciso$. Corelations with respect to this factorisation system are simply cospans.
With this setup, $\catC$ satisfies Assumption~\ref{ass:thcorelations}, and Theorem~\ref{thm:corelations} then states that 
\begin{equation}\label{eq:cubeisos}
\vcenter{
    \xymatrix@C=40pt{
      {\Ciso \+{\Ciso} \op{\Ciso}} \ar[r] \ar[d] & {\Spanc{\Ciso}}
      \ar[d]^{\SpanToCorel} \\
    {\Cospan{\catC}} \ar[r]_-{\CospanToCorel} & {\Cospan{\catC}} \pullbackcorner
    }
   }
\end{equation}
is a pushout square. Within the axiomatic perspective of Corollary \ref{cor:presentation}, \eqref{eq:cubeisos}
states that all the axioms arising from pullbacks of isomorphisms in $\catC$
already arise from pushouts in $\catC$. This is a manifestation of the fact that
any square in $\catC$ formed by taking the pullback of isomorphisms is also a
pushout square.

\subsection{Lattices} \label{ex:latt}
Another ``degenerate'' case is the one of lattices. Recall that a lattice is a
poset with binary joins and meets. Just as for any pre-order, one can see a
lattice $L$ as a category: objects are the lattice elements and there is an
arrow from $a \in L$ to $b \in L$ precisely when $a \leq b$ in the lattice. 

Let us now check Assumption~\ref{ass:thcorelations} for $L$. First, $L$ has
pushouts and pullbacks: they are just given by joins and meets. For the second
and third requirement, we need to fix a factorisation system. Above we noticed
that the factorisation system $(\catC,\Ciso)$ works for any $\catC$; in the case
of a lattice, we may also consider the factorisation $(\Ciso,\catC)$. Indeed,
because there is at most one arrow between any two objects of $L$, every arrow
is both an epi and a mono, and thus every subcategory of $L$ is a subcategory of
the monos in $L$. Also, notice that the only isomorphisms are the identities;
thus the subcategory $\Ciso$ of isomorphisms is just the discrete sub-category $|L|$ formed by the objects of $L$, with their identity arrows. Choosing the factorisation system
$(|L|,L)$, and letting $\subc =L$, then gives data that satisfies
Assumption~\ref{ass:thcorelations}.

(Co)spans for a lattice turn out to be familiar order-theoretic notions. The set of arrows from $a \in L$ to $b \in L$ in $\Cospan{L}$ is equal to the \emph{upper
set} of $a \vee b$, and composition is given by join. Similarly, the set of arrows from $a \in L$ to $b \in L$ in $\Span{L}$ is equal to the \emph{lower set} of $a \wedge b$, and
composition is given by meet. Then the category $\Corel{L}$ is equivalent to the
terminal category $\Triv$; any two cospans in $L$ are considered equivalent as
$(| L |,L)$-corelations.

Theorem~\ref{thm:corelations} thus yields the pushout square
\[
    \xymatrix@C=40pt{
      {L \+{L} \op{L}} \ar[r] \ar[d] & {\Spanc{L}}
      \ar[d] \\
    {\Cospan{L}} \ar[r] & {\Triv.} \pullbackcorner
    }
\]
This collapse is reminiscent of Non-Example~\ref{ex:triv}. We can understand
this as follows. A morphism from $a$ to $b$ in $L\+{L} \op{L}$ is a zigzagging
path from $a$ to $b$ in the lattice $L$, for example $a \le c_1 \ge c_2 \le c_3
\ge \dots \le c_n \ge b$. The category $\Cospan{L}$ considers two paths to be
the same if they have the same least upper bound, while the category $\Spanc{L}$
considers two paths to be the same if they have the same greatest lower bound.
Pushing out these two notions of equivalence, we see that in the category
$\Corel{L}$, all paths between $a$ and $b$ are considered the same, and hence
there is a unique morphism between any two objects. That is to say, $\Corel{L}
\cong \Triv$.

Our example can be generalised from lattices to arbitrary posets with pushouts
and pullbacks. It is straightforward to show that a poset $P$ with all pushouts
and pullbacks need not itself be a lattice, but that elements $a$ and $b$ have a
join if and only if they have a meet, and moreover if and only if they have an
upper bound or lower bound. This means that $P$ must be a coproduct of lattices. 

Suppose that $P$ is the coproduct of $n$ lattices. In this case, we may then use
the same (identity, any morphism) factorisation system $(|P|,P)$ and apply
Theorem~\ref{thm:corelations}. Here the homsets of $\Cospan{P}$ and $\Span{P}$
can still be described as upper sets of joins and lower sets of meets, with the
caveat that the homset is empty when the relevant join or meet does not exist.
The category of $(|P|,P)$-corelations is then equivalent to the coproduct of $n$
copies of the terminal category, and Theorem~\ref{thm:corelations} asserts the usual pushout
square. Thus in this case the pushout of the categories of spans and cospans in
$P$ results in counting the number of lattice summands, or `connected
components', of $P$.

\subsection{Linear Subspaces over a Field} \label{ex:sv}

We now consider an example for the abelian case: the prop $\SV{\field}$ whose arrows $n \to m$ are $\field$-linear subspaces of $\field^n \times \field^m$, for a field $\field$. Composition in $\SV{\field}$ is relational: $V \poi W = \{(v,w) \mid \exists u .(v,u) \in V, (u,w) \in W\}$. Interest in $\SV{\field}$ is motivated by various recent applications. We mention the case where $\field$ is the field of Laurent series, in which $\SV{\field}$ constitutes a denotational semantics for signal flow graphs~\cite{Bonchi2014b,BaezErbele-CategoriesInControl,Bonchi2015,BonchiSZ17}, and the case $\field = \Z_2$, in which $\SV{\field}$ is isomorphic to the phase-free ZX-calculus, an algebra for quantum observables~\cite{CoeckeDuncanZX2011,BialgAreFrob14}.

Now, in order to apply our construction, note that $\SV{\field}$ is isomorphic to $\Rel{\Vect{k}}$, where $\Vect{k}$ is the abelian prop whose arrows $n \to m$ are the linear maps of type $\field^n \to \field^m$ (the monoidal product is by direct sum). This follows from the observation that subspaces of $\field^n \times \field^m$ of dimension $z$ correspond to mono linear maps from $\field^z$ to $\field^n \times \field^m$, whence to jointly mono spans $n \tl{} z \tr{}m$ in $\Vect{k}$.

We are then in position to use Corollary \ref{thm:corelationsAbPROP}, which yields the following pushout characterisation for $\SV{\field}$.
\[
\begin{aligned}
    \xymatrix@R=18pt{
    {\Vect{k} + \op{\Vect{k}}} \ar[r]
    \ar[d]& {\Span{\Vect{k}}}
    \ar[d] \\
    {\Cospan{\Vect{k}}} \ar[r] & {\SV{\field}} \pullbackcorner
    }
  \end{aligned}
\]
This very same pushout has been studied in~\cite{BialgAreFrob14} for the $\field = \Z_2$ case. As before, the modular reconstruction suggests a presentation by generators and relations for $\SV{\field}$, in terms of the theories for spans and cospans in $\Vect{k}$. The axiomatisation of $\SV{\field}$ is called the theory of interacting Hopf algebras~\cite{interactinghopf,ZanasiThesis} , as it features two Hopf algebras structures and axioms expressing their combination.

On the top of existing results on $\SV{\field}$, our Corollary \ref{thm:corelationsAbPROP} suggests a
novel perspective, namely that $\SV{\field}$ can be also thought as the prop of
\emph{corelations} over $\Vect{\field}$. This representation can be understood
by recalling the 1-1 correspondence between subspaces of $\field^n \times
\field^m$ and (solution sets of) homogeneous systems of equations $M v = 0$,
where $M$ is a $z \times (n+m)$ matrix. Writing the block decomposition $M =
(M_1\:\vert\: -M_2)$, where $M_1$ is a $z \times n$ matrix and $M_2$ a $z \times
m$ matrix, this is the same as solutions to $M_1v_1 = M_2v_2$. These systems
then yield jointly epi cospans $n \tr{M_1} z \tl{M_2} m$ in $\Vect{\field}$,
that is, corelations.

\subsection{Linear Corelations over a Principal Ideal Domain} \label{ex:PID}

We now consider the generalisation of the linear case from fields to principal
ideal domains (PIDs). In order to form a prop, we need to restrict our attention
to finitely-dimensional \emph{free} modules over a PID $\PID$. The symmetric
monoidal category of such modules and module homomorphisms, with monoidal
product by direct sum, is equivalent to the prop $\fmod\PID$ whose arrows $n \to
m$ are $\PID$-module homomorphisms $\PID^n \to \PID^m$ or, equivalently,
$m\times n$-matrices in $\PID$. Because of the restriction to free modules,
$\fmod\PID$ is not abelian. However, it is still finitely bicomplete and has a
costable (epi, split mono)-factorisation system.\footnote{The factorisation
given by (epi, mono) morphisms is not unique up to isomorphism, whence the
restriction to split monos---see~\cite{Fong2015}.} Note that the fact that the ring $\PID$ is a PID matters for the existence of pullbacks, as it is necessary for
submodules of free $\PID$-modules to be free---pushouts exist by self-duality
of $\fmod\PID$.

Write $\mfmod\PID$ for the prop of split monos in $\fmod\PID$. It is a
classical, although nontrivial, theorem in control theory that this category
obeys the required condition on pushouts of pullbacks \cite{Fong2015}. Hence
Theorem~\ref{thm:corelationsPROPs} yields the pushout square
\begin{equation}
  \label{eq:pushoutCorelMod1}
  \begin{aligned}
    \xymatrix@C=40pt{
      \fmod\PID + \op{\fmod \PID} \ar[d] \ar[r] & \Span{\mfmod\PID} \ar[d], \\
      \Cospan{\fmod \PID} \ar[r] &\Corel{\fmod\PID} \pullbackcorner
    }
  \end{aligned}
\end{equation}
in $\PROP$. This modular account of $\Corel{\fmod\PID}$ is relevant for the semantics of dynamical systems. When $\PID = \mathbb{R}[s,s^{-1}]$, the ring of Laurent polynomials
in some formal symbol $s$ with coefficients in the reals, the prop
$\Corel{\fmod{\mathbb{R}[s,s^{-1}]}}$ models complete linear time-invariant
discrete-time dynamical systems in $\mathbb{R}$; more details can be found
in~\cite{Fong2015}. In that paper, it is also proven that $\Corel{\fmod\PID}$ is axiomatised by the presentation of $\Cospan{\fmod \PID}$ with the addition of the law $\Bunit\poi\Bcounit \ = \ \EmptyDiag$. By Corollary~\ref{cor:presentation}, it follows that $\Bunit\poi\Bcounit \ = \ \EmptyDiag$ originates by a pullback in $\Span{\mfmod\PID}$ and in this case it is the only contribution of spans to the presentation of corelations.
\smallskip

It is worth noticing that, even though $\fmod\PID$ is not abelian, the pushout of spans and cospans over $\fmod\PID$ does not have a trivial outcome as for the prop $\F$ of functions (Example~\ref{ex:triv}). Instead, in~\cite{interactinghopf,ZanasiThesis} it is proven that we have the pushout square
\begin{equation}
  \label{eq:pushoutCorelMod2}
  \begin{aligned}
    \xymatrix@C=40pt{
      \fmod{\PID} \oplus \op{\fmod \PID} \ar[d] \ar[r] & \Span{\fmod \PID} \ar[d], \\
      \Cospan{\fmod \PID} \ar[r] &\SV\frPID \pullbackcorner
    }
  \end{aligned}
\end{equation}
in $\PROP$, where $\frPID$ is the \emph{field of fractions} of $\PID$. 

The pushout \eqref{eq:pushoutCorelMod2} is relevant for the categorical
semantics for signal flow graphs pursued in \cite{Bonchi2014b,Bonchi2015,BonchiSZ17}. Even though it is not
an instance of Theorem~\ref{thm:corelationsPROPs} or Corollary \ref{thm:corelationsAbPROP}, our developments shed light on \eqref{eq:pushoutCorelMod2} through the comparison with \eqref{eq:pushoutCorelMod1}. First, note that any element $r \in \PID$ yields a module homomorphism $x \mapsto rx$ in $\fmod \PID$ of type $1 \to 1$, represented as a string diagram $\scalar$. The key observation is that, in \eqref{eq:pushoutCorelMod2}, $\Span{\fmod \PID}$ is contributing to the axiomatisation of $\SV\frPID$ (\emph{cf.} Corollary \ref{cor:presentation}) by adding, for each $r$, an equation 
\begin{equation*}
  \scalar \poi \scalarop \ = \ \IdDiag \poi \IdDiag \text{, corresponding to a pullback}
  \vcenter{\xymatrix@=10pt{
    & 1 & \\
    1\ar[ur]^{\scalar} && 1 \ar[ul]_{\scalar} \\
    & \ar[ul]^{\IdDiag} 1 \ar[ur]_{\IdDiag} &
  }}
\end{equation*}
in $\fmod\PID$. Back to \eqref{eq:pushoutCorelMod1}, the only equations of this kind that $\Span{\mfmod \PID}$ is contributing with are those in which $\scalar$ is a \emph{split mono}, that means, when $r$ is \emph{invertible} in $\PID$. Therefore, the difference between \eqref{eq:pushoutCorelMod1} and \eqref{eq:pushoutCorelMod2} is that in the latter one is adding formal inverses $\scalarop$ also for elements $\scalar$ which are not originally invertible in $\PID$. This explains the need of the field of fractions $\frPID$ of $\PID$ in expanding the pushout object from $\Corel{\fmod\PID}$ (in  \eqref{eq:pushoutCorelMod1}) to $\Corel{\Vect{\frPID}} \cong \SV\frPID$ (in  \eqref{eq:pushoutCorelMod2}).

\subsection{Relations of Algebras over a Field}
\label{ssec.monadalg}

We give a simple example of Proposition~\ref{prop.algebras}. The category
$\vect\field$ of vector spaces and linear maps over some field $\field$, is a
complete, cocomplete, regular category. Moreover, every epi splits: every
surjective linear map has a section. 

Consider the monad $T\colon \vect\field \to \vect\field$ for algebras over
$\field$ (that is, for vector spaces equipped with a bilinear multiplication).
This maps a vector space $V$ to the direct sum $\bigoplus_{\{n \in \mathbb{N}\}}
V^{\otimes n}$. Since each map $V \mapsto V\otimes \dots \otimes V$ preserves
both epis and pushouts of epis, so does their coproduct $T$. Note also that all
epis in $\vect\field$ are regular. This means both $T$ and $T^2$ preserve
pushouts of regular epis, and hence Proposition~\ref{prop.algebras} applies.
Thus we can construct the category of relations between algebras over $\field$
as a pushout of spans and cospans. 

\subsection{Corelations of Coalgebras}
\label{ssec.comonadcoalg}

As mentioned, Proposition \ref{prop:algebrasmonad} dualises to yield a universal
construction for corelations of coalgebras for a comonad $T$. The dualised
assumptions require that $T$ and $T^2$ preserve pullbacks of monos. This is true
for instance when $T$ is a polynomial endofunctor, as it is the case for
comonads encountered in the setting of functional programming like the
environment comonad $A \mapsto A \times E$, see \cite{TarmoComonads}. Another
example is when $T$ is the comonad defined by a geometric morphism to a coregular
topos, such as $\SET$.

An interesting class of case studies for this construction stems from the analysis of state-based systems, which are typically formalised as coalgebras $X \to FX$ for an endofunctor $F \colon \catC \to \catC$~\cite{JacobsBookCoalgebra}. When $F$ is accessible and $\catC$ is locally finitely presentable, one can construct the cofree comonad $T \colon \catC \to \catC$ over $F$ through a so-called terminal sequence~\cite{Worrell99}. The comonad $T$ yields an operational semantics for $F$-systems: given $X \to FX$, it maps $X$ to the carrier $\Omega$ of the cofree $F$-coalgebra $\Omega \to F\Omega$ on $X$. 

Now, because $T$ is cofree, the condition of preserving pullbacks of monos required in Proposition~\ref{prop:algebrasmonad} can actually be checked on $F$. Indeed, the image under $T$ of a pullback diagram can be computed pointwise on its terminal sequence, which essentially involves repeated applications of the functor $F$. The condition is met for instance for coalgebras for the finite double-powerset functor $\mathcal{PP} \colon \SET \to \SET$: this functor has been used to model ground logic programs \cite{KomendantskayaMP10,KomendantskayaP16,BZ-LogProg2}, with operational semantics given by the comonad mapping a formula to its execution tree in the logic program. 

Given that bisimulation of coalgebras can be expressed both in terms of spans and cospans \cite{StatonBisimulation}, it is an interesting question  --- to be explored in future work --- whether corelations also model a sensible notion of behavioural equivalence, for well-chosen examples. 

\section{Concluding remarks} \label{sec:conclusion}

In summary, we have shown that categories of (co)relations may, under certain
general conditions, be constructed as pushouts of categories of spans and
cospans.  In particular, especially since categories of spans and cospans can
frequently be axiomatised using distributive laws, this offers a method of
constructing axiomatisations of categories of (co)relations.  Our results extend
to the setting of props, and more generally symmetric monoidal categories.
Moreover, these results are readily illustrated, unifying a diverse series of
examples drawn from algebraic theories, program semantics, quantum computation,
and control theory.

Looking forward, note that in the monoidal case the resulting (co)relation
category is a so-named hypergraph category: each object is equipped with a
special commutative Frobenius structure.  Hypergraph categories are of
increasing interest for modelling network-style diagrammatic languages, and
recent work, such as that of decorated corelations \cite{Fon16} or the
generalized relations of Marsden and Genovese \cite{MG17}, gives precise methods for
tailoring constructions of these categories towards chosen applications. Our
example on relations in categories of algebras for a monad
(Subsection~\ref{ssec.monadalg}) hints at general methods for showing the
present universal construction applies to these novel examples. We leave this as
an avenue for future work.

\bibliographystyle{plainurl}
\bibliography{lmcsBib}

\newpage \appendix

%\section{Omitted Proofs}

\section{Proof of Theorem \ref{thm:corelations}}\label{sec:proof}

We devote this section to give a step-by-step argument for Theorem \ref{thm:corelations}. 

\begin{prop}\label{prop:pushoutcommutes} The square~\eqref{eq:pushoutCorel} commutes. \end{prop}
\begin{proof}
  As ${\subc \+{\subc} \op{\subc}}$ is a pushout, it is enough to show
  \eqref{eq:pushoutCorel} commutes on the two injections of $\subc$, $\op{\subc}$ into ${\subc
  \+{\subc} \op{\subc}}$. This means that we have to show, for any $f \colon
  a\to b$ in $\subc$, that
  \[
    \SpanToCorel(\tl{\id}\tr{f}) = \CospanToCorel(\tr{f}\tl{\id}) \quad\text{ and }\quad
    \SpanToCorel(\tl{f}\tr{\id}) = \CospanToCorel(\tr{\id}\tl{f}).
  \]
  These are symmetric, so it suffices to check one.  This follows immediately
  from the fact that the pushout of $\tl{\id}\tr{f}$ is $\tr{f}\tl{\id}$.
\end{proof}

Suppose we have a cocone over $\Cospan\catC \longleftarrow \subc\+{\subc}\subc
\longrightarrow \Spanc\subc$.  That is, suppose we have the commutative square:
\begin{equation}
\label{eq:arbitrary}
\tag{$\dag$}
\raise15pt\hbox{$
\xymatrix@C=40pt{
{\subc \+{\subc} \op{\subc}} \ar[r] \ar[d] & {\Spanc{\subc}} \ar[d]^{\Psi} \\
{\Cospan{\catC}} \ar[r]_-{\Phi} & {\mathcal{X}}.
}$}
\end{equation}
We prove two lemmas, from which the main theorem follows easily.
\begin{lem} \label{lemma:arbitraryCommutativeDiag} 
  If $\tr{p} \tl{q} $ is a cospan in $\subc$ with pullback $\tl{f} \tr{g}$ (in
  $\catC$), then $\Psi(\tl{f} \tr{g}) = \Phi(\tr{p} \tl{q} )$. Similarly, if
  $ \tl{f} \tr{g} $ is a span in $\subc$ with pushout $\tr{p}\tl{q}$ (in
  $\catC$), then $\Psi(\tl{f} \tr{g} ) = \Phi(\tr{p} \tl{q})$.  
\end{lem}
\begin{proof}
  Consider $\tr{p}\tl{q} \in \subc \+{\subc} \op{\subc}$.  Its image in
  $\mathcal X$ via the lower left corner of the commutative square
  \eqref{eq:arbitrary} is $\Phi(\tr{p}\tl{q})$ while, recalling that
  $\tr{p}\tl{q}  \in \subc \+{\subc} \op{\subc}$ is mapped to $\tl{f}\tr{g}$ in $\Span{\subc}$, its image via the upper right
  corner is $\Psi(\tl{f}\tr{g})$. Thus $\Phi(\tr{p}\tl{q}) = \Psi(\tl{f}\tr{g})$.

  The second claim is analogous, beginning instead with the span
  $\tl{f}\tr{g}$.
\end{proof}

\begin{lem} \label{lemma:functordescendstocorel}
  If $\tr{p_1}\tl{q_1}$ and $\tr{p_2}\tl{q_2}$
  are cospans in $\catC$ such that $\CospanToCorel(\tr{p_1}
  \tl{q_1})=\CospanToCorel(\tr{p_2} \tl{q_2})$, then $\Phi(\tr{p_1} \tl{q_1})
  = \Phi(\tr{p_2} \tl{q_2})$.
\end{lem}
\begin{proof}
Suppose $\CospanToCorel(\tr{p_1}\tl{q_1}) =
\CospanToCorel(\tr{p_2}\tl{q_2})$ as per hypothesis. Then by Proposition \ref{lemma:charCospanToCorel} there exists
$\tr{m_1},\tr{m_2} \in \Cmono$ and $\tr{p}\tl{q} \in \Cospan\catC$ such that
\[
  \tr{p_1}\tl{q_1} \ = \ \tr{p}\tr{m_1}\tl{m_1}\tl{q} \quad\mbox{and}\quad
\tr{p_2}\tl{q_2} \ = \ \tr{p}\tr{m_2}\tl{m_2}\tl{q}.
\]
Then
\begin{eqnarray*}
\Phi(\tr{p_1}\tl{q_1}) &=& \Phi(\tr{p}\tr{m_1}\tl{m_1}\tl{q})
\\ &=& \Phi(\tr{p}\tl{id})\poi\Phi(\tr{m_1}\tl{m_1})\poi\Phi(\tr{id}\tl{q})
\\ &\overset{(\clubsuit)}=&
\Phi(\tr{p}\tl{id})\poi\Psi(\tl{\id}\tr{\id})\poi\Phi(\tr{\id}\tl{q})
 \\  &=& \Phi(\tr{p}\tl{id})\poi\Phi(\tr{\id}\tl{q})
 \\  &=& \Phi(\tr{p}\tl{q}),
\end{eqnarray*}
and similarly for $\tr{p_2}\tl{q_2}$.
The equality $(\clubsuit)$ holds because, by Assumption~\ref{ass:thcorelations}, $\Cmono \subseteq \subc$ and $\tr{m_1 \in \Cmono}$ is mono, thus the pullback of
$\tr{m_1}\tl{m_1}$ is $\tl{\id}\tr{\id}$ and via Lemma~\ref{lemma:arbitraryCommutativeDiag} $\Phi(\tr{m_1}\tl{m_1}) = \Psi(\tl{\id}\tr{\id})$. \end{proof}

\begin{proof}[Proof of Theorem~\ref{thm:corelations}]
Suppose we have a commutative
diagram \eqref{eq:arbitrary}. It suffices to show that there exists a
functor $\theta \colon \Corel{\catC}\to\mathcal{X}$ with
$\theta\CospanToCorel=\Phi$ and $\theta\SpanToCorel =
\Psi$. Uniqueness is automatic by fullness (Proposition~\ref{prop:CospanToCorelFull})
and bijectivity on objects of~$\CospanToCorel$.

Given a corelation $a$, fullness yields a cospan $\tr{f}\tl{g}$ such that $\CospanToCorel(\tr{f}\tl{g}) = a$. We then define $\theta(a) = \Phi(\tr{f}\tl{g})$. This is well-defined by Lemma \ref{lemma:functordescendstocorel}.

For commutativity, clearly $\theta\CospanToCorel = \Phi$.  Moreover, $\theta\SpanToCorel =
\Psi$: given a span $\tl{f}\tr{g}$ in $\Cmono$, let $\tr{p}\tl{q}$ be its
pushout span in $\catC$. Thus by
Lemma~\ref{lemma:arbitraryCommutativeDiag},
\[
  \Psi(\tl{f}\tr{g}) = \Phi(\tr{p}\tl{q})
  =\theta\CospanToCorel(\tr{p}\tl{q})=\theta\SpanToCorel(\tl{f}\tr{g}).
  \tag*{\qedhere}
\]
\end{proof}

\section{Proof of Proposition~\ref{thm:corelationsSMCs}}\label{app:spansPROP}

To begin, we discuss how to put monoidal structures on $\Cospan{\catC}$,
$\Corel{\catC}$, and $\Spanc{\subc}$, and show that $\CospanToCorel$ and
$\SpanToCorel$ are strict symmetric monoidal functors in this case.

  \begin{prop}
    Let $(\catC,\tns)$ be a symmetric monoidal category with pullbacks, and let
    $\subc$ be a sub-symmetric monoidal category of $\catC$ containing all
    isomorphisms and stable under pullback. If $\tns$ preserves pullbacks in
    $\subc$, then $(\Spanc{\subc},\tns)$ is a symmetric monoidal category.
  \end{prop}
  \begin{proof}
    Define $\tns\maps \Spanc{\subc} \times \Spanc{\subc} \longrightarrow \Spanc{\subc}$ to take a pair of objects $(X,X')$ in $\Spanc{\subc}$ to $X \tns X'$, and
    similarly map a pair of spans to the component-wise monoidal product. This
    is well defined since $\subc$ is closed under $\tns$. We first need to show
    that this map is functorial. That is, given two pairs $(X \leftarrow N \to Y,\:X' \leftarrow N' \to Y')$ and
    $(Y \leftarrow M \to Z,\: Y' \leftarrow M' \to Z')$ of spans in $\subc$, we need
    to show that the composite of their images under $\tns$:
    \[
      X\tns X' \longleftarrow (N\tns N') \times_{Y\tns Y'} (M\tns M')
      \longrightarrow Z \tns Z'
    \]
    is isomorphic to the image under $\tns$ of their composite:
    \[
      X \tns X' \longleftarrow (N\times_YM) \tns (N'\times_{Y'}M')
      \longrightarrow Z \tns Z'.
    \]
    This is precisely the hypothesis that $\oplus$ preserves pullbacks in
    $\subc$.

    It then remains to show that we have coherence maps, and these obey the
    requisite equations. Note that $\subc$ contains all isomorphisms in $\catC$,
    and hence the coherence maps for $(\catC,\tns)$ can be considered as
    morphisms in $\Spanc{\subc}$. It is easy to check these give the coherence
    of $\Spanc{\subc}$.
  \end{proof}
 
Note that dualising the above argument with $\subc = \catC$ yields the fact that
$(\Cospan{\catC},\tns)$ is a symmetric monoidal category whenever $\tns$
preseves pushouts.  Also note that the inclusions $\subc \to \Spanc{\subc}$ and
$\op{\subc} \to \Spanc{\subc}$ are strict monoidal functors.

\begin{prop}
  If $\catC$ is a symmetric monoidal category with a costable factorisation
  system, and $\Cmono$ is closed under $\tns$, then $\Corel{\catC}$ is a
  symmetric monoidal category. Moreover, the quotient functor $\CospanToCorel\maps \Cospan{\catC} \to \Corel{\catC}$ is a strict monoidal functor.
\end{prop}
\begin{proof}
  The first task is to show that $\Corel{\catC}$ is indeed a symmetric monoidal
  category. We show that $\tns$ induces a monoidal product, which we shall also
  write $\tns$, on $\Corel{\catC}$. Given two corelations $a$ and $b$, with
  representatives $\tr{f}\tl{g}$ and $\tr{h}\tl{k}$ we define their monoidal
  product $a \tns b$ to be the corelation represented by the cospan $\tr{f\tns
  h} \tl{g \tns k}$.  This is well defined: given $\tr{f'}\tl{g'}$,
  $\tr{h'}\tl{k'}$ and $m_1,m_2$ in $\Cmono$ such that $f'= f;m_1$, $g' =
  g;m_1$, $h'=h;m_2$, $k' = k;m_2$, the monoidality of $\tns$ in $\catC$ implies
  $f'\tns g' = (f\tns g);(m_1 \tns m_2)$ and $h'\tns k' = (h\tns k);(m_1 \tns
  m_2)$. Since $\Cmono$ is closed under $\tns$, $m_1 \oplus m_2$ again lies in
  $\tns$, and the product corelation is independent of choice of
  representatives.
  
  To show that $\CospanToCorel$ is a strict monoidal functor we just need to
  check $\CospanToCorel(a \tns b) =\CospanToCorel a \tns \CospanToCorel b$,
  where $a$ and $b$ are cospans. This follows immediately from the definition:
  the monoidal product of the corelations that two cospans represent is by
  definition the corelation represented by the monoidal product of the two
  cospans.
\end{proof}

\begin{prop}
  $\SpanToCorel\maps \Spanc{\subc} \to \Corel{\catC}$ is a strict monoidal
  functor.
\end{prop}
\begin{proof}
  Again, we just need to check that $\SpanToCorel(a\tns b) = \SpanToCorel a \tns
  \SpanToCorel b$. For this we need that the monoidal product preserves pushouts. 
  Indeed, given spans $a = X \tl{f} N \tr{g} Y$ and $b = X' \tl{f'} N'
  \tr{g'} Y'$, we have $\SpanToCorel(a\tns b)$ represented by the cospan 
  \[
    X\tns X' \longrightarrow (X\tns X')+_{(N \tns N')}(Y\tns Y')
    \longleftarrow Y \tns Y',
  \]
  and $\SpanToCorel a \tns \SpanToCorel b$ represented by the cospan
  \[
    X\tns X' \longrightarrow (X+_{N}Y)\tns (X'+_{N'}Y')
    \longleftarrow Y \tns Y'.
  \]
  These cospans are isomorphic by the fact $\tns$ preserves pushouts, and
  hence represent the same corelation.
\end{proof}

Now having described how to interpret \eqref{eq:pushoutCorel} in a category of
symmetric monoidal categories, it remains to show that it is a pushout.

\begin{proof}[Proof of Proposition \ref{thm:corelationsSMCs}]
  From Theorem~\ref{thm:corelations}, we know that \eqref{eq:pushoutCorel}
  commutes, and that given some other cocone
  \[
    \raise15pt\hbox{$
      \xymatrix@C=40pt{
	{\subc \+{\catC} \op{\subc}} \ar[r] \ar[d] & {\Spanc{\subc}} \ar[d]^{\Psi} \\
	{\Cospan{\catC}} \ar[r]_-{\Phi} & {\mathcal{X}},
      }$}
  \]
  there exists a unique functor $\theta$ from $\Corel{\catC}$ to $\mathcal{X}$.
  All we need do here is check that $\theta$ is a (lax, strong, strict) monoidal
  functor whenever $\Phi$ and $\Psi$ are (lax, strong, strict) monoidal
  functors. In fact, whether $\theta$ is lax, strong, or strict depends only on
  whether $\Phi$ is lax, strong, or strict.

  Indeed, suppose we have corelations $a\colon X \to Y$ and $b\colon X'\to Y'$,
  and write $\tilde a$ and $\tilde b$ for cospans that represent them. Recall
  that by definition $\theta a = \Phi \tilde a$. Let the coherence maps
  $\phi_{X,X'}\colon \Phi X \tns \Phi X' \to \Phi(X\tns X')$ of $\Phi$ also be
  the coherence maps of $\theta$. Then the coherence of $\Phi$ states that
  \[
    \xymatrix@C=40pt{
      \Phi X \tns \Phi X' \ar[r]^{\phi_{X,X'}} \ar[d]_{\theta a \tns \theta b =
      \Phi \tilde a \tns \Phi \tilde b}
      & \Phi (X \tns X') \ar[d]^{\theta(a\tns b) = \Phi(\tilde
      a \tns \tilde b)} \\
      \Phi Y \tns \Phi Y' \ar[r]^{\phi_{Y,Y'}} & \Phi (Y \tns Y')
      }
  \]
  commutes. This shows that $\theta$ is a monoidal functor of the same type as
  $\Phi$, and hence proves the theorem.
\end{proof}

\end{document}